\def\ready{}

\ifdefined\ready
\documentclass[sigplan,nonacm]{acmart}
\else
\documentclass[sigplan,anonymous,review,nonacm]{acmart}
\fi

\usepackage{preamble}

\newcommand{\numBenchmarks}{35}

\newcommand{\numSuiteQubitsLo}{12}

\newcommand{\numSuiteQubitsHi}{100}

\newcommand{\numSuiteTGatesMaxK}{186}

\newcommand{\numDascotCoverage}{16}

\newcommand{\numMenuDensestTiles}{1.1}

\newcommand{\numMenuFastestSpaceMultiple}{3}

\newcommand{\numMenuSlowdownDensestConfig}{24}

\newcommand{\numSoleShare}{86}

\newcommand{\numSoleCount}{30}

\newcommand{\numTightChipRushhour}{29}

\newcommand{\numTightChipStatic}{13}

\newcommand{\numSuiteBudgetRushhour}{1.8}

\newcommand{\numSuiteBudgetBaseline}{2.1}

\newcommand{\numMinChipLo}{1.2}

\newcommand{\numMinChipHi}{3.5}

\newcommand{\numSpeedupTight}{7.2}

\newcommand{\numSpeedupBandLo}{2.3}

\newcommand{\numSpeedupBandHi}{7.2}

\newcommand{\numOutrightBestShare}{94}

\newcommand{\numMaxGapToBest}{1.2}

\newcommand{\numMatchedPuremagic}{0.9}

\newcommand{\numMatchedLo}{2.2}

\newcommand{\numMatchedHi}{9.2}

\newcommand{\numPbcLeadOthreels}{2.0}

\newcommand{\numPbcLeadPuremagic}{2.0}

\newcommand{\numPbcLeadCountPuremagic}{26}

\newcommand{\numPbcLeadLitinski}{2.1}

\newcommand{\numPuremagicFloorRatio}{0.64}

\newcommand{\numNoPbcReplay}{6.1}

\newcommand{\numQubitsMatchedLo}{1.3}

\newcommand{\numQubitsMatchedHi}{5.8}

\newcommand{\numQubitSecondsLo}{2.8}

\newcommand{\numQubitSecondsHi}{13}

\newcommand{\numLogicalVolumeBest}{1.4}

\newcommand{\numLogicalVolumeNext}{2.1}

\newcommand{\numQubitSecondsAdvantage}{2.3}

\newcommand{\numQubitSecondsAdvantageMax}{10}

\newcommand{\numBoundSpaceRushhour}{1.21}

\newcommand{\numBoundTimeRushhour}{3.99}

\newcommand{\numBoundTimeLsqca}{44}

\newcommand{\numBoundSpacePuremagic}{1.78}

\newcommand{\numBoundTimePuremagic}{2.67}

\newcommand{\numBoundGapGeomean}{4.8}

\newcommand{\numBoundGapMax}{9.8}

\newcommand{\numBoundGapOthersLo}{13}

\newcommand{\numBoundGapOthersHi}{55}

\newcommand{\numFrontierRushhour}{22}

\newcommand{\numFrontierBaselineLo}{2}

\newcommand{\numFrontierBaselineHi}{5}

\newcommand{\numFamilies}{13}

\newcommand{\numFamiliesCheapest}{8}

\newcommand{\numFamiliesWithinTwoPercent}{2}

\newcommand{\numFrozenRemoved}{46}

\newcommand{\numFrozenRemovedMax}{98}

\newcommand{\numFrozenSurvivesTight}{9}

\newcommand{\numFloorHeadroom}{1.1}

\newcommand{\numFloorTightLo}{3.0}

\newcommand{\numFloorTightHi}{4.9}

\newcommand{\numFloorStaticPooled}{26}

\newcommand{\numTeleportedShare}{97}

\newcommand{\numCultivationsLo}{1.5}

\newcommand{\numCultivationsHi}{8}

\newcommand{\numErrorIdleMax}{88}

\newcommand{\numErrorMagicTight}{11}

\newcommand{\numErrorMergeTight}{2}

\newcommand{\numErrorMergeBound}{5}

\newcommand{\numErrorMotionBound}{0.2}

\newcommand{\numDynamismGatesPerReshapeTight}{140}

\newcommand{\numDynamismHadamardsPerRotationTight}{1{,}000}

\newcommand{\numDynamismMotionShareTight}{0.14}

\newcommand{\numSensitivityFastMatchedLo}{1.2}

\newcommand{\numSensitivityFastMatchedHi}{10.9}

\newcommand{\numSensitivitySlowMatchedLo}{0.8}

\newcommand{\numSensitivitySlowMatchedHi}{5.4}

\newcommand{\numSensitivityBoundGap}{4.8}

\newcommand{\numSensitivityBoundGapDoubleMotion}{5.2}

\newcommand{\numSensitivityWalkingShift}{2.2}

\newcommand{\numPolicyTeleportTimeMax}{1.6}

\newcommand{\numPolicyTeleportReshapeMax}{613}

\newcommand{\numPolicyLookaheadReshapeMax}{10}

\newcommand{\numPolicyPlacementMax}{1.3}

\newcommand{\numSweepCompiles}{966}

\newcommand{\numSweepCpuMinutes}{64}

\newcommand{\numCompileMedianSeconds}{2.2}

\newcommand{\numEpsilonGrantMedian}{12}

\newcommand{\numTacoFairnessFloor}{1.77}

\begin{document}
\title{RushHour: A Dynamically Reconfigurable Lattice-Surgery Architecture}

\author{Nathaniel Tornow, Aleksandra \'{S}wierkowska, Peter Wegmann, Pramod Bhatotia}
\affiliation{\institution{Technical University of Munich}\country{}}

\begin{abstract}
Practical fault-tolerant quantum computing (FTQC) requires efficient lattice surgery (\LS), so that large algorithms fit on resource-constrained quantum chips.
Existing approaches, however, are rigid: qubits, routing space, and resource states are allocated ahead of execution, which prevents running on small chips, leaves statically scheduled executions with large time overheads, and fixes each design at a single area of the space--time trade-off.

We present \textit{dynamic \LS}, which enables efficient reconfiguration of the ancilla space, just-in-time allocation of resource states, and dynamic rotations of logical qubits, thereby spanning the entire space--time trade-off with a single, unified approach.
We realize dynamic \LS with \rush{} through a hardware-compiler co-design: the \emph{\rush{} ISA} formalizes and programs our dynamic lattice model, the \emph{Lattice Management Unit} abstracts dynamic lattice management and performs efficient lattice reconfiguration, and the \emph{\rush{} Compiler} compiles logical circuits for physical chips into optimized ISA programs while pipelining instructions.

We evaluate \rush{} against six state-of-the-art compilers and two resource models.
On the smallest chips, $\numSoleShare\%$ of benchmarks run only with \rush{}, while existing approaches require $\numMinChipLo$--$\numMinChipHi\times$ larger chips.
On space-constrained early-FTQC chips, \rush{} runs a median $\numSpeedupBandLo$--$\numSpeedupBandHi\times$ faster than the best feasible alternative, while matching the state of the art on very large chips.
\rush{}'s constructive results run $\numBoundGapGeomean\times$ from an idealized-machine resource limit.
\end{abstract}

\maketitle

\section{Introduction}
\label{sec:introduction}

\begin{figure}[t]
    \centering
    \includegraphics[width=\columnwidth]{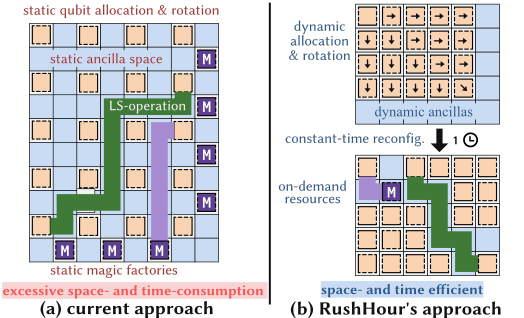}
    \caption{\textbf{Static vs. dynamic lattice surgery.} (a) Current architectures fix qubit allocation, ancilla space, and magic-state factories ahead of execution. (b) \rush{} allocates, rotates, and reconfigures dynamically and on demand.}
    \label{fig:comparison}
\end{figure}

Realizing the potential of quantum computers requires quantum error correction (QEC) to protect computations against the high error rates inherent in physical quantum hardware~\cite{acharya2025below}.
The leading approach to QEC is the \textit{surface code}, which encodes logical qubits into two-dimensional patches of noisy physical qubits and can suppress errors exponentially in code size~\cite{fowler2012surface,acharya2025below}. However, even a single logical qubit may require hundreds or thousands of physical qubits, making efficient use of the available hardware essential~\cite{beverland2022assessing}.

Logical operations in the surface code are implemented using \textit{lattice surgery} (\LS)~\cite{litinski2019game} by temporarily merging and splitting neighboring patches to perform joint logical measurements. \LS operations require additional \textit{ancilla patches} to connect data patches and produce resource states~\cite{litinski2019game,fowler2018lowoverhead}.

The size and arrangement of this ancilla space determine whether and how efficiently a fault-tolerant program can execute~\cite{flasq}. With limited ancilla space, a computation can fit on a smaller chip, but logical operations must be serialized. With more ancilla space, operations can proceed in parallel, reducing execution time at the cost of a larger physical-qubit footprint that might not fit a physical chip.

An \LS architecture must therefore balance two competing objectives:
\begin{inparaenum}[(a)]
\item \textit{space efficiency}, to execute large quantum circuits on devices with a constrained number of physical qubits, and
\item \textit{time efficiency}, to minimize exposure to logical errors and complete the computation as quickly as possible.
\end{inparaenum}

However, current \LS architectures are \textbf{fundamentally limited} due to the \textbf{static} allocation and configuration of the lattice \textit{ahead of execution} (\figref{comparison})~\cite{litinski2019game,molavi2025dependency,watkins2024high,zhu2026o3ls,kobori2025lsqca,hofmeyr2025scheduling,wang2024taco}:

\noindent
\textbf{First}, static approaches cannot feasibly run on small physical chips, as they must statically reserve ancillas to ensure every gate remains routable (\figref{limits})~\cite{kobori2025lsqca,watkins2024high,zhu2026o3ls,litinski2019game}.

\textbf{Second}, the current approach imposes a large, unnecessary time overhead: \LS{} operations serialize as gates queue at statically placed factories, fixed compute regions, or routing buses, while other free patches remain idle~\cite{litinski2019game,kobori2025lsqca}.
Each qubit's patch orientation is also fixed before execution, so an operation needing a specific boundary of a qubit may first pay an expensive patch rotation~\cite{litinski2019game,zhu2026o3ls}.

\textbf{Third}, no existing approach is both space- and time-efficient, and each occupies only a small, suboptimal area of the space--time trade-off, as current approaches build on floor plans fixed before execution, whether hand-designed or selected per workload~\cite{ghosh2026workload}. The densest approach uses $\numMenuDensestTiles$ tiles per qubit but runs a median $\numMenuSlowdownDensestConfig\times$ slower than the fastest approach, which in turn requires over $\numMenuFastestSpaceMultiple\times$ the space.

We argue that overcoming these limitations requires the lattice to be managed dynamically as the computation unfolds. This leads to our main research question: 

\researchquestion{Research Question}{
    How can we design an \LS architecture that is dynamically configurable during execution to achieve space- and time-efficient FTQC across all chip sizes?
}
\medskip

Our \textit{dynamic \LS} model (\figref{comparison}) enables efficient ancilla relocation within a single round, so that a few ancilla resources can be reused across operations, reducing the required number of ancillas and, consequently, the required chip size by $\numMinChipLo$--$\numMinChipHi\times$.
On-demand allocation of resource states enables ready gates to execute without waiting for statically reserved resources, while patch
orientations become a virtually free state of the dynamic lattice, enabling time-efficient execution of the circuit.
More broadly, dynamic lattice surgery provides a unified \LS{} model across the space--time trade-off: under tight area constraints, it reuses lattice tiles over time, whereas any additional area can be exploited to expose greater parallelism.

\begin{figure}[t]
    \centering
    \includegraphics[width=\columnwidth]{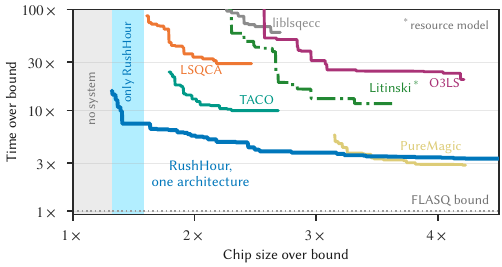}
    \caption{\textbf{Best runtime by chip size}, normalized by the
    FLASQ bound~\cite{flasq}.
    }
    \vspace{-3mm}
    \label{fig:limits}
\end{figure}

\begin{figure*}[t]
    \centering
    \includegraphics[width=\textwidth]{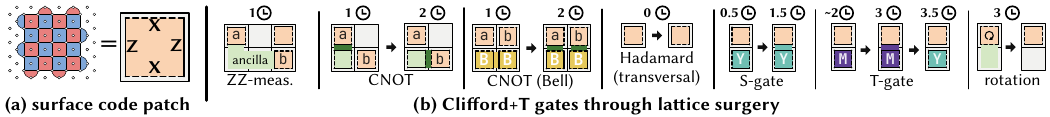}
    \caption{\textbf{Surface-code patches and \LS{} operations}~\cite{litinski2019game,flasq}. All times are in logical rounds. \texttt{M}: a cultivated magic state~\cite{gidney2024magic} (occupying its tile for $2$--$3$ expected rounds), \texttt{Y}: the $|Y\rangle$ ancilla~\cite{gidney2023inplace}. Dashed edges mark the $X$ boundary.}
    \label{fig:ls}
\end{figure*}

\myparagraph{Technical challenges}
Realizing dynamic lattice surgery presents several challenges:
First, we must efficiently represent and carefully reason about a lattice whose geometric configuration and allocations can change over time, thereby providing a compilation target for valid instructions.
Second, dynamic reconfiguration must not become a performance bottleneck: both the latency of each reconfiguration and the total number of reconfigurations must be minimized to prevent dynamic execution from stalling the computation.
Third, dynamic reconfigurability and the flexibility to allocate resources at runtime substantially increase compilation complexity, further complicating an already NP-hard optimization problem~\cite{herr2017lattice,siraichi2018qubit}.
Finally, dynamic reconfigurations and allocations could introduce unnecessary latency if not scheduled carefully off the critical execution path.

\myparagraph{Our approach: Architecture-compiler co-design}
To address these challenges, we present \rush{}, a dynamically reconfigurable \LS{} architecture comprising three core components:
\begin{inparaenum}[(1)]
\item The \textit{\rush{} ISA} formalizes dynamic \LS{} by exposing lattice geometry and resource allocations as machine state, specifying legal state transitions and available gate operations to the compiler, and enabling non-blocking single-round reconfigurations by construction.
\item The \textit{Lattice Management Unit} (\amu{}) maintains and abstracts this state by providing a minimal interface for realizing gates and reshaping the lattice when the current configuration cannot support them. Its mechanisms are designed to minimize both reconfiguration latency and frequency.
\item The \textit{\rush{} compiler} schedules operations, reconfigurations, and resource allocations by pipelining operations to shift their overhead off the critical path, and searches for Pareto-optimal configurations that balance execution time and resource usage for a given circuit and chip.
\end{inparaenum}

\myparagraph{Results}
We implement \rush{} and evaluate it on \numBenchmarks{} representative benchmarks against six competing architectures and compilers, and two resource models (\secref{evaluation}).
At \rush{}'s minimum feasible chip, $\numSoleShare\%$ of benchmarks run only with \rush{}, while existing approaches require $\numMinChipLo$--$\numMinChipHi\times$ larger chips.
On space-constrained chips, \rush{} reaches a successful shot up to a median $\numSpeedupTight\times$ sooner than the best feasible alternative and, at $2\times$ the physical minimum, delivers the outright best result on $\numOutrightBestShare\%$ of benchmarks.
In total space--time resource cost, \rush{} runs $\numQubitSecondsLo$--$\numQubitSecondsHi\times$ cheaper than every design except the spacious-chips-only PureMagic, with which it is on par.

\myparagraph{Contributions} We make the following contributions:
\begin{enumerate}[leftmargin=5mm]
    \itemsep0em
    \item We introduce \emph{dynamic \LS{}}, an execution model that makes the lattice a dynamic machine state. A dynamic ancilla corridor reconfigures in a single round by construction, resource states are allocated just in time, and qubit orientations become a tracked property of the lattice (\secref{machine}).
    \item We realize dynamic \LS{} through the \rush{} ISA, which exposes configurations, transitions, and gate realizations as a compilation target, and the Lattice Management Unit (\amu{}), which abstracts the complexity and efficiently manages the dynamic lattice (\secref{machine}, \secref{amu}).
    \item We build the \rush{} compiler, which schedules gates, dynamic reconfigurations, and state preparations while pipelining their overhead off the critical path, thereby producing a circuit's Pareto-optimal space--time trade-off for a given physical chip (\secref{compiler}).
\end{enumerate} %
\section{Background}
\label{sec:background}

\subsection{Quantum Error Correction with the Surface Code}
Useful quantum computing requires quantum error correction (QEC) to suppress the errors of physical qubits~\cite{acharya2025below}.
The surface code is the leading QEC code for two-dimensional nearest-neighbor hardware: a distance-$d$ patch encodes one logical qubit in a $d{\times}d$ array of data qubits with interleaved $X$- and $Z$-stabilizer ancillas, occupying $2(d{+}1)^2$ physical qubits per tile~\cite{fowler2012surface,flasq}.
Its logical error rate per cycle falls exponentially with distance $d$ below threshold~\cite{beverland2022assessing}, so raising $d$ trades physical qubits for exponentially lower error.
We abstract a patch as one tile of the chip (\figref{ls}a). Its four edges carry two $Z$- and two $X$-boundaries on opposite pairs, which define the logical operators and how patches interact~\cite{litinski2019game}.

\subsection{Lattice Surgery and Gates}

We implement logical fault-tolerant operations on the surface code using \textit{lattice surgery} (\LS), which enables two-qubit Pauli-product measurements of $XX$ and $ZZ$ operators via the merging and splitting of surface code patches~\cite{horsman2012lattice,litinski2019game,fowler2018lowoverhead}. Each logical merge takes the time of $d$ surface-code cycles of duration $t_{\mathrm{cyc}} \approx 1\,\mu\mathrm{s}$ each on superconducting hardware~\cite{flasq}. Throughout, one \emph{logical round} is $d$ QEC cycles (${\sim}d\,\mu$s). A merge takes 1 logical round, and splitting takes 0 rounds~\cite{litinski2019game}.
As the example in \figref{ls}b shows, we can perform a $ZZ$-measurement between two qubits $a$ and $b$ by merging and splitting them through ancilla space between the two qubits.
Using \LS{}, we can implement the Clifford+$T$ gate set:

A \CNOT{} gate is implemented as a sequence of two Pauli-product measurements on an ancilla qubit $c$ prepared in $\ket{+}$: a $ZZ$-measurement between the control and $c$, followed by an $XX$-measurement between $c$ and the target, which takes 2 logical rounds in total~\cite{horsman2012lattice,litinski2019game}.
A \CNOT{} can also execute through a prepared Bell pair of qubits $b_1, b_2$: a $ZZ$ merge of the control with $b_1$ and an $XX$ merge of $b_2$ with the target act on different patches through disjoint regions, so both run concurrently, and the gate completes in one logical round instead of two, given that the pair is prepared ahead of time~\cite{litinski2019game,fowler2018lowoverhead}.
The outcomes leave only Pauli byproducts that are tracked in software~\cite{litinski2019game}.

A Hadamard gate is applied transversally in 0 logical rounds. It exchanges the patch's $X$- and $Z$-boundaries~\cite{horsman2012lattice}, and restoring the original orientation requires a patch rotation through an adjacent ancilla tile in 3 logical rounds~\cite{litinski2019game}.

An $S$/$S^\dagger$ gate is implemented via a $ZZ$-measurement on an ancilla qubit in the $\ket{Y}$ state, which can be prepared in $0.5$ rounds~\cite{gidney2023inplace, flasq}.

A $T$ gate consumes a magic state $\ket{m} = \ket{0} + e^{i\pi/4}\ket{1}$ through a $ZZ$-measurement. Depending on the measurement outcome, we must apply a conditional $S$-gate correction~\cite{bravyi2005universal}.
With magic-state cultivation~\cite{gidney2024magic}, we can produce $\ket{m}$ in a single ancilla tile in approximately $2$ logical rounds.

\subsection{Walking Qubits}
\label{subsec:walking-qubits}
Walking qubits allow us to slide surface-code patches across the lattice in straight or diagonal directions by shifting the stabilizer-measurement schedule by one lattice site per QEC cycle~\cite{mcewen2023relaxing}, a primitive already demonstrated in hardware~\cite{eickbusch2025dynamic}. McEwen et al.~\cite{mcewen2023relaxing} establish walking at one patch width per two logical rounds~\cite{flasq}, while a single-round schedule is conjectured from the $Y$-basis construction of~\cite{gidney2023inplace} but not yet demonstrated as a fault-tolerant protocol~\cite{flasq}. We assume the one-round rate (one patch width per logical round) throughout, while also reporting every aggregate at the proven two-round rate as a sensitivity arm (\tabref{sensitivity}), where no result moves by more than $\numSensitivityWalkingShift\%$.

\figref{walking-qubits} shows the primitive sliding many qubit patches in parallel to reconfigure ancilla space. We show two examples of reconfiguring the blue ancilla space into distinct shapes within a single round. This parallel qubit sliding forms the core mechanism for efficient dynamic reconfiguration in this work.

\begin{figure}[t]
    \centering
    \includegraphics[width=1\linewidth]{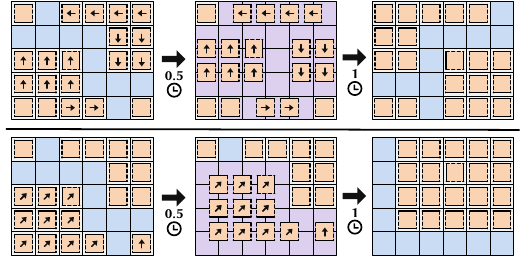}
    \caption{\textbf{Walking qubits.} Parallel patch slides reshape the free
    space in one round~\cite{mcewen2023relaxing,flasq}.}
    \label{fig:walking-qubits}
\end{figure}

Because the slide is realized by ordinary (reconfigured) stabilizer cycles rather than a logical operation, it runs at the resting error rate and implements the logical identity~\cite{mcewen2023relaxing}.

\section{\rush{} Overview}
\label{sec:overview}

\figref{overview} shows an overview of \rush{}, which we divide into three main components:

(1) The \textbf{\rush{} ISA} models the dynamic lattice as a mutable machine state with valid configurations, reconfigurations, allocation transitions, and gate realizations, serving as our execution model. 
Our dynamic lattice enables fast and rare reconfigurations by construction, since dynamic ancillas are arranged as a connected staircase corridor between patches, making any valid corridor shape one round of walking away from another, regardless of distance.

(2) The \textbf{Lattice Management Unit (\amu{})} abstracts the complexity of realizing circuits on the dynamic lattice through a simple interface.
To do so, the translation table stores the current state of the dynamic lattice, the access engine computes \LS{}-realizations of logical gates under the current lattice configuration, and the reconfiguration engine efficiently reshapes the lattice when a gate cannot execute in the current state of the lattice, doing so in a single round while serving as many gates as possible.

(3) To compute an optimized realization of a given circuit on a given physical chip using the dynamic lattice model, the \textbf{\rush{} Compiler} efficiently schedules transitions and gates while pipelining dynamic reshapes and state preparations with other gates.
Because a single compilation is cheap, we compile, for the target chip, a sweep of candidate regions and placements it admits, price each at every code distance it affords, and return the dominant operating points.

\begin{figure}
    \centering
    \includegraphics[width=1\linewidth]{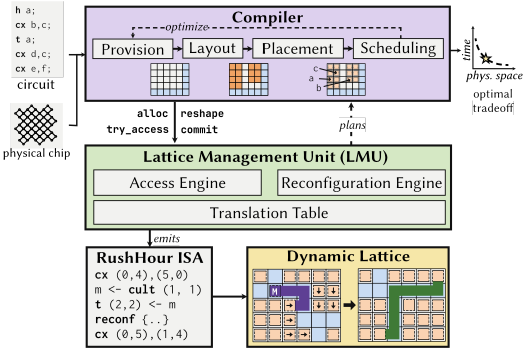}
    \caption{\textbf{\rush{} overview.} The compiler (\secref{compiler}) schedules against the \amu{} interface (\secref{amu}). The \amu{} manages the dynamic lattice and realizes gates as \rush{} ISA instructions (\secref{machine}).}
    \label{fig:overview}
\end{figure}

\begin{figure*}[t]
    \centering
    \includegraphics[width=1\textwidth]{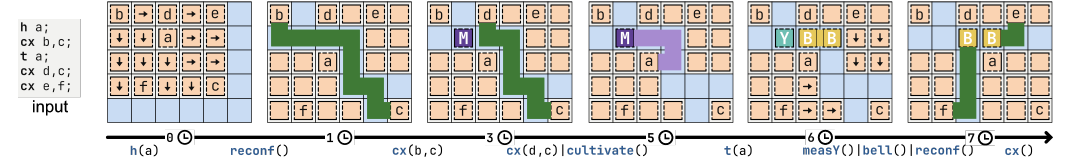}
    \caption{\textbf{Example \rush{} ISA stream.}}
    \label{fig:isa-example}
\end{figure*}

\begin{figure}[t]
    \centering
    \includegraphics[width=\linewidth]{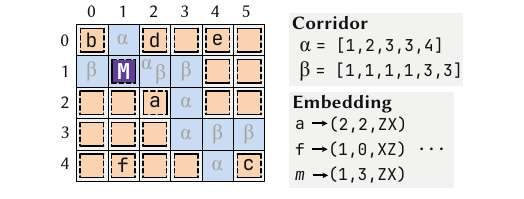}
    \caption{\textbf{Dynamic-lattice configuration $\Sigma$.} The arrays $\alpha,\beta$ place $n+m+1$ tiles, joining into a staircase corridor (blue). The embedding maps each qubit to its slot and orientation.}
    \label{fig:lattice}
\end{figure}

\myparagraph{Example}
\figref{isa-example} shows the compiled \rush{} ISA stream for the logical circuit on the left.
The \texttt{h(a)} gate executes transversally in 0 rounds and flips the boundary orientation of \texttt{a}.
A single-round reconfiguration then reshapes the blue corridor into a configuration that enables the next three gates: \texttt{cx(b,c)} executes in two rounds via the green path through the new corridor, followed by \texttt{cx(d,c)} and, concurrently, an in-place magic-state preparation~\cite{gidney2024magic} (here, two rounds), enabling \texttt{t(a)} via a $ZZ$ measurement in one round, followed by a conditional $Y$-basis correction.
The correction runs in parallel with a second reconfiguration and a Bell pair, enabling the final \texttt{cx(e,f)} to complete in a single round.

\section{Dynamic Lattice Surgery and the \rush{} ISA}
\label{sec:machine}

To implement dynamic lattice surgery, we treat lattice geometry as a mutable machine state.
\rush{} enables this dynamic lattice model through three core properties: (i) the validity of a configuration is a local, checkable condition, (ii) any valid configuration can reach any other in a single one-round transition, and (iii) any free tile can host any resource, allocated dynamically and uniformly across the lattice.

The \rush{} ISA (\figref{ir}) describes all instructions that make up a program on the dynamic lattice, including \emph{transitions}, which can change a configuration, and \emph{gates}, which execute on the standing configuration.

\subsection{The Dynamic Lattice}
\label{subsec:model}
The dynamic lattice is an $n{\times}m$ grid $\region$ of data patches embedded in an $(n{+}1){\times}(m{+}1)$ chip grid $\grid$ leaving $n{+}m{+}1$ free tiles that join into one \emph{staircase corridor} winding between the patches (\figref{lattice}).
We distinguish \emph{slots} from \emph{tiles}: slots index the virtual $n{\times}m$ data grid $\region$, tiles the physical $(n{+}1){\times}(m{+}1)$ chip grid $\grid$.

The machine state is a \emph{configuration}
\begin{equation*}
\Sigma \;=\; (E,\ \alpha,\ \beta,\ A,\ \tau).
\end{equation*}
\begin{itemize}[leftmargin=5mm]
    \itemsep0em
    \item $E$: the \emph{embedding}, mapping each logical qubit to its slot and its boundary orientation $r \in \{\textsc{xz},\textsc{zx}\}$ (\figref{ir}),
    \item $\alpha, \beta$: the \emph{corridor arrays} locating the free space, $\alpha[i]$ being the column where the corridor crosses row $i$, $\beta[j]$ the row where it crosses column $j$,
    \item $A$: the \emph{allocation}, storing transient occupants including cultivated magic states, $|Y\rangle$ states, and prepared Bell pairs,
    \item $\tau$: the \emph{reservations}, recording the round at which each occupied tile frees again.
\end{itemize}

\myparagraph{Displacement}
A patch sits at its home tile unless the corridor has crossed past it on an axis, which displaces it by one tile along that axis, so every patch sits within one tile of its home slot (\figref{lattice}).

\myparagraph{Validity}
A configuration is \emph{valid} iff $(\alpha, \beta)$ are \emph{jointly monotone}: both arrays are nondecreasing, and each unit step of one is matched by a crossing of the other, so the crossings interlock (\figref{lattice} shows a valid pair).
Joint monotonicity ensures that the free tiles chain into a single connected corridor, and every patch sits within one tile of its home slot, in accordance with the displacement rule above.

\myparagraph{Example}
\figref{lattice} shows the configuration of the running example (\figref{isa-example}) at round 3.
The corridor arrays pin one free tile per row and column ($\beta[3] = 1$ leaves tile $(1,3)$ free, $\alpha[2] = 3$ tile $(2,3)$), and the embedding records each qubit's slot and orientation, as shown for \texttt{a} and \texttt{f}, while $A$ holds transient occupants such as the magic state \texttt{\textit{m}}.

\myparagraph{Larger regions}
A chip larger than the $n{\times}m$ patches leaves slots that no qubit binds and that stay permanently free, forming interior free lanes beside the corridor. Gates use this free space for dynamic resource allocation and execution without relying on the corridor, enabling a fluid trade-off between space footprint and time overhead through reconfiguration~(\secref{compiler}).

\subsection{Reconfiguration}
\label{subsec:reconf}
A reconfiguration moves the corridor wherever a gate's operands need it, represented by \isa{reconf} operations that take a set of tiles and the direction each tile's patch should walk.

A reconfiguration between two valid configurations is guaranteed to complete in a single round. \figref{walking-qubits} gives the intuition behind the following lemma:
\begin{lemma}[One-walk reachability]
\label{lem:one-slide}
Any valid configuration lies a single one-tile slide per patch away from any other.
\end{lemma}
\noindent
\emph{Proof sketch.}
Each patch's offsets are read off $(\alpha, \beta)$, so a reconfiguration slides every patch by its offset difference: at most one tile per axis, however far the corridor moves.
Such slides are legal provided each patch stays on the chip and no two patches land on one tile, cross head-on, or sweep a standing patch's corner.
Joint monotonicity never moves same-line neighbors toward each other, so their gap never falls below one tile. The case of two diagonal neighbors crossing one $2{\times}2$ block is guarded by the interlocking crossings that keep the corridor's free tile inside it, so the pair crosses around that corner without any contact. A diagonal slide across a $2{\times}2$ block with a standing patch at its corner is excluded the same way: the interlocking crossings keep the block's free tile between the mover and the standing patch, so the swept corner is free.
No slide waits for another, so the reconfiguration is one round of parallel slides.

\begin{figure}[t]\footnotesize\centering
\setlength{\arraycolsep}{2pt}
\resizebox{\columnwidth}{!}{$
\begin{array}{r@{\;}c@{\;}l@{\quad}l@{\quad}r}
\multicolumn{5}{@{}l}{\textsc{Logical IR}\ \opnote{(Clifford${+}T$)}}\\[1pt]
\multicolumn{5}{@{}l}{\rule{0pt}{1pt}}\\[-6pt]
g &\Coloneqq& \op{H}\,q \mid \op{X}\,q \mid \op{Z}\,q
       \mid \op{S}^{\pm}q \mid \op{T}^{\pm}q
       \mid \op{CX}\,q_c,q_t & & \\[7pt]
\multicolumn{4}{@{}l}{\textsc{\rush{} ISA}} & \opnote{rounds}\\[1pt]
\Pi &\Coloneqq& o^{*} & \opnote{program} & \\
o &\Coloneqq& \op{alloc}\;q@s\!:\!r & \opnote{patch lifecycle} & 0\\
  &\mid& \op{H}\,p \mid \op{X}\,p \mid \op{Z}\,p & \opnote{in-place Clifford} & 0\\
  &\mid& \op{cultivate}\;p & \opnote{magic state} & v_{\mathrm{cult}}(d)^{\ast}\\
  &\mid& \op{prepY}\;p \mid \op{measY}\;y & \opnote{$|Y\rangle$ access} & 0.5\\
  &\mid& \op{reconf}\;\{(p,\delta)^{*}\} & \opnote{slide (one round)} & 1\\
  &\mid& \op{rotate}\;p\;\opvia\;a & \opnote{orientation flip} & 3\\
  &\mid& \op{bell}\;(b_1,r_a),(b_2,r_b)\;\opvia\;\rho & \opnote{Bell pair} & 1\\
  &\mid& \op{cx}\;p_c,p_t \;\opvia\;\rho & \opnote{CX gate} & 2\\
  &\mid& \op{cx}\;p_c,p_t \gets (b_1,b_2) \;\opvia\;\rho_c,\rho_t & \opnote{teleported} & 1\\
  &\mid& \op{T}^{\pm}\;p \gets m \;\opvia\;\rho & \opnote{magic merge} & 1\\
  &\mid& \op{S}^{\pm}\;p \gets y \;\opvia\;\rho & \opnote{$|Y\rangle$ merge} & 1
\end{array}
$}
\caption{\textbf{\rush{} ISA.} The logical Clifford+$T$ input and the ISA the
compiler emits, each operation priced in logical rounds ($^{\ast}$expected
cultivation occupancy, \secref{methodology})~\cite{litinski2019game,gidney2023inplace,flasq}. Variables: logical qubits $q \in \Qset$, data slots $s \in \region$, tile
coordinates $p,m,y,b \in \mathbb{Z}^2$, patch orientation
$r \in \{\textsc{xz},\textsc{zx}\}$, unit slide step $\delta \in \textsc{Dir}_8$,
and path tiles $\rho \subseteq \mathbb{Z}^2$.}
\label{fig:ir}
\end{figure}

\subsection{On-demand State Allocation}
\label{subsec:alloc}
Any free tile in the data region or corridor can host resource states, including magic states, $|Y\rangle$ states, and Bell pairs, created where needed and consumed via a gate instruction.

Because a resource lives in a free tile, it can be prepared ahead of the consuming gate. E.g., a Bell pair can be laid down on idle tiles early, letting the teleported \CNOT{} complete in a single round once its data operands are free (\figref{isa-example}).

\subsection{Instructions}
\label{subsec:isa}
Each state-transition and gate becomes one instruction of \figref{ir}, with the tiles it uses as arguments and a given duration in logical rounds:
\begin{inparaenum}[(1)]
    \item a reconfiguration lowers to \isa{reconf},
    \item an on-demand allocation to \isa{cultivate}, \isa{prepY}, or \isa{bell},
    \item each remaining gate to the merge that consumes its operands.
\end{inparaenum}
The Hadamard gate is a special case, since it is logically a gate, but lowers into a pure transition, costing zero rounds and touching no tile, reducing only to an edit of the orientation $r$ in the embedding $E$.
The lattice catches up on this rotation lazily, as an actual \isa{rotate} only needs to be emitted when a later gate demands a specific boundary exposed to ancilla space.

An instruction runs as soon as its tiles are free, so independent instructions overlap automatically.
Deterministically re-executing the stream while tracking the configuration $\Sigma$ therefore yields the makespan, making the dynamic lattice a target the compiler schedules against directly.      %
\section{The Lattice Management Unit}
\label{sec:amu}

\begin{figure}[t]
    \centering
    \includegraphics[width=\linewidth]{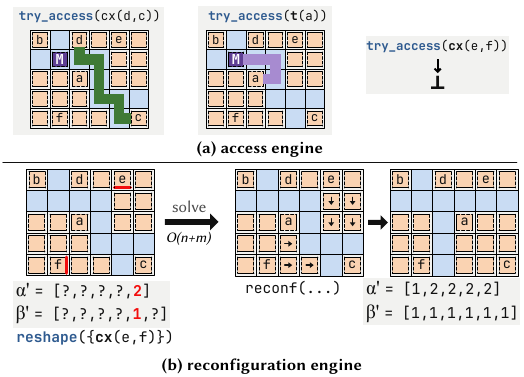}
    \caption{\textbf{\amu{} engines.} (a) The access engine answers
    \amucall{try\_access} with realizations on the lattice.
    (b) A refusal enters the reconfiguration engine: the operation \emph{pins}
    corridor entries (red), and the solver completes them to a full corridor in $\mathcal{O}(n+m)$ time.}
    \label{fig:lmu}
\end{figure}

The Lattice Management Unit (\amu{}) abstracts the complexity of the dynamic lattice to provide \LS realizations of logical gates and efficient reconfigurations through its translation table, access engine, and reconfiguration engine.

\subsection{\amu{} Interface}

The \amu{} exposes a set of four main functions:

\noindent
\textbf{--\amucall{alloc}} allocates a logical qubit onto a virtual slot, registering the qubit in the translation table.

\noindent
\textbf{--\amucall{try\_access}} takes a logical gate $g$ and the set $\claimed$ of claimed tiles and returns the realizations of $g$ on the current configuration that avoid $\claimed$, or $\bot$ if none exist.

\noindent
\textbf{--\amucall{reshape}} takes the current gate frontier $\frontier$ of ready gates and computes a reconfiguration that serves its most critical blocked gate, greedily batching further blocked gates into the same solve.

\noindent
\textbf{--\amucall{commit}} emits ISA operations into the final instruction stream.

\subsection{Access Engine}
The access engine implements the \amucall{try\_access} function to return realizations of a given logical gate $g$.
To do so, it first looks up the physical locations and orientations of the logical qubit operands in the translation table.
Then, it enumerates possible realizations of the gate using breadth-first search, finding valid paths through unoccupied space in the lattice.
\figref{lmu}~(a) shows three example calls for the lattice state of \figref{lattice}. Access for \texttt{cx(d,c)} and for \texttt{t(a)} each returns a single realization with the shown ancilla paths.
Access for \texttt{cx(e,f)} returns $\bot$: the standing configuration offers no realization because no corridor connects \texttt{e} and \texttt{f}.

The access engine also determines where to allocate resource states.
For a $T$ or $S$ gate, it collects up to eight candidate tiles in nearest-first order through the free space, starting from the consumer's $Z$-boundary, and picks the one that allows the gate to start earliest.
Bell pairs for teleported \CNOT{}s follow the same principle: candidate sites nearest the operands' boundary entries are scored by earliest consumption.

\subsection{Reconfiguration Engine}

The reconfiguration engine implements the \amucall{reshape} function to enable the corridor to be reconfigured to another location in a single round.
\figref{lmu}~(b) shows an example of a reshape between two configurations, enabling the previously infeasible \texttt{cx(e,f)}.

\myparagraph{Resolving the corridor}
A blocked operation requires free ancillas at specific boundaries of its operands, each demand fixing one entry of the completion as a \textit{pin} (red).
The corridor solver takes such a pin set and either extends it to a full valid completion $(\alpha', \beta')$ if one exists or reports $\bot$.

To do so, it maintains one feasible interval per entry of $\alpha'$, derived from joint monotonicity, tightens the intervals by every pinned crossing of either array, and propagates the bounds forward and backward. The array $\beta'$ is then constructed from the completed $\alpha'$, and an interval that empties or a construction that fails proves the pin set infeasible.
The solver therefore determines a legal corridor for a requested set in $O(n+m)$ time if one exists.

\begin{algorithm}[t]
\caption{\rush{} scheduling algorithm.}
\label{alg:compiler}
\small
\begin{algorithmic}[1]
\Require circuit $\Circ$; placement $\placement$ onto an $n \times m$ data region $\region$
\Ensure certified physical-op stream
\State \amucall{alloc}$(\qbit, \placement(\qbit))$ for all $\qbit$
\State $\frontier \gets$ ready gates of $\Circ$
\While{$\frontier \neq \emptyset$}
  \State $\claimed \gets \emptyset$ \Comment{tiles claimed in this scan}
  \ForAll{$\gate \in \frontier$ by criticality}
    \State $\plans \gets{}$\amucall{try\_access}$(\gate, \claimed)$
    \If{$\plans \neq \bot$}
      \State $\bestplan \gets \arg\min_{\plan \in \plans} \tdone(\plan)$
      \State \amucall{commit}$(\bestplan)$
      \State $\claimed \gets \claimed \cup \tilesof{\bestplan}$
      \State advance $\frontier$
    \EndIf
  \EndFor
  \If{nothing was served}
    \amucall{commit}(\amucall{reshape}$(\frontier)$)
  \EndIf
\EndWhile
\State \Return the emitted stream, certified against $\Circ$ by replay
\end{algorithmic}
\end{algorithm}

\myparagraph{Reshaping}
When no gate of the current gate frontier can be served on the standing configuration $\Sigma$, \amucall{reshape} computes a resolving \isa{reconf} instruction in four steps (with a rotation fallback):

\emph{(1)~Pin.}
For the most critical blocked gate $\gate$, each candidate realization determines the corridor entries it requires, yielding a constant number of pin sets.

\emph{(2)~Solve.}
The corridor solver extends every pin set to a target configuration or refutes it.

\emph{(3)~Batch.}
Each surviving configuration greedily takes on further blocked gates of $\frontier$, growing its pin set one gate at a time and re-solving, keeping an extension only if every batched gate still routes.

\emph{(4)~Select.}
The candidate serving the most gates wins. Ties break toward the fewest moved patches, and the winner is lowered into a single \isa{reconf} instruction.
A reshape thus issues one $\mathcal{O}(n{+}m)$ solve per pin set and one re-solve per batched gate, so serving a gate frontier costs $\mathcal{O}(|\frontier|)$ solves.

Should no candidate route $\gate$, a fifth step flips one operand's orientation, appends a \isa{rotate}, and redoes the solve.

\section{The \rush{} Compiler}
\label{sec:compiler}

The \rush{} compiler transforms a logical Clifford+$T$ circuit into a valid stream of \rush{} ISA instructions that execute in minimal time on a given physical chip.

To do so, we operate on two levels.
Within one \emph{candidate}, a data region $\region$ paired with a placement $\placement$, a scheduling loop serves the circuit through the \amu{} (\secref{compiler:scheduling}).
Across candidates, a sweep compiles the candidate regions and placements the chip admits, prices each at every code distance the chip affords, and returns the Pareto-dominant operating points (\secref{compiler:sweep}).

\subsection{Workflow}
One compile realizes one candidate in four passes:

\emph{(1)~Provision.}
We lay out the dynamic lattice as the $(n{+}1) \times (m{+}1)$ tile grid $\grid$ that hosts the candidate's data region and its corridor (\subsecref{model}).

\emph{(2)~Layout.}
We carve the candidate's data region $\region$ of slots out of $\grid$ (\subsecref{model}), and a corridor, spreading any surplus area evenly into interior free lanes.

\emph{(3)~Placement.}
We bind each logical qubit to a slot of $\region$ under the candidate's placement strategy $\placement$ (\secref{compiler:placement}).

\emph{(4)~Scheduling.}
Finally, we serve the gates by lowering each to its \rush{} ISA realization. For this, we query the \amu{} for the ready gates the standing lattice can execute, commit the cheapest realization, and reshape the lattice on blocking gates (\secref{compiler:scheduling}).

\subsection{Placement}
\label{sec:compiler:placement}

The placement pass computes an initial embedding $E$ of the configuration $\Sigma$ (\subsecref{model}).
We implement three strategies:

With \emph{locality-first} placement, we weight qubit pairs by their \CNOT{} count and map the qubits by a snake-like traversal of the lattice onto the slots, allowing interacting qubits to stay adjacent so that their merges route through short corridors.
With \emph{density-first} placement, we pack qubits with the most ancilla-requiring gates at the border of free tiles left open by the layout pass.
With \emph{spreading} placement, we distribute the qubits uniformly over the lattice without ordering.

\subsection{Scheduling}
\label{sec:compiler:scheduling}

The scheduler issues \rush{} ISA instructions while minimizing the makespan of the instructions (Algorithm~\ref{alg:compiler}).
It first binds every qubit to its slot with \amucall{alloc} derived from placement and initializes the gate frontier $\frontier$ of ready gates.
Every iteration either serves gates or reshapes the lattice to make them feasible, until the gate frontier $\frontier$ empties:

Each scan iterates over the gate frontier by criticality, where a gate's criticality is the length of the longest chain of rounds that depends on it.
Per gate, \amucall{try\_access} returns the realizations avoiding the tiles $\claimed$ already claimed in this scan. The scheduler commits the earliest-finishing realization $\bestplan$, adds its tiles to $\claimed$, and advances $\frontier$.
A committed operation issues the moment its tiles free up, so operations overlap, and preparations pipeline beneath running gates.

When a scan cannot serve any gates with the current configuration, the scheduler calls \amucall{reshape} on $\frontier$ and commits the returned reconfiguration that serves as many gates in the frontier as possible~(\secref{amu}).

\subsection{Candidate Sweep}
\label{sec:compiler:sweep}

One run of the four passes produces one ISA stream for the candidate data region $\region$ and placement $\placement$.
Given a chip, the sweep computes the circuit's Pareto frontier in three steps:

\emph{(1)~Enumerate.}
We sample the regions $\region$ from two shape families: near-square regions on a geometric area ladder starting at the register minimum $|\Qset|$ and stretched regions with a short side of up to ten slots. Every region is paired with the three placements.
The free-area ceiling is bounded by the circuit's peak number of concurrent gates, each requiring a corridor of typical length $\sqrt{|\Qset|}$.

\emph{(2)~Compile and price.}
Each candidate is compiled once per cultivation occupancy class, and its stream is priced at every code distance $d$ for which the patch grid fits the chip. Each (stream, $d$) pair yields one operating point (physical qubits, $T_{\mathrm{succ}}$, \secref{methodology}).

\emph{(3)~Select.}
The sweep forms the circuit's frontier by keeping the Pareto-dominant feasible operating points across all candidates.
When the frontier is still improving at the free-area ceiling, the sweep raises this ceiling until the frontier flattens.

\section{Experimental Methodology}
\label{sec:methodology}

\myparagraph{Setup and benchmarks}
We implement \rush{} in ${\sim}9{,}300$ lines of Rust and evaluate it on \numBenchmarks{} representative Clifford+$T$ circuits from FTCircuitBench~\cite{ftbench} and MQT Bench~\cite{quetschlich2023mqtbench} at \numSuiteQubitsLo--\numSuiteQubitsHi{} qubits and up to $\numSuiteTGatesMaxK$k $T$ gates. A breakdown of the benchmark circuits is provided in Appendix~\ref{app:per-benchmark}.
All compilations run on a 384-core server with a one-hour timeout and 4\,GB of memory per circuit.

\myparagraph{Baselines}
We compare \rush{} against the compilers of LSQCA~\cite{kobori2025lsqca}, DASCOT~\cite{molavi2025dependency}, the Lattice Surgery Compiler (lsqecc, referred to as liblsqecc below)~\cite{watkins2024high,leblond2024realistic}, O3LS~\cite{zhu2026o3ls}, PureMagic~\cite{hofmeyr2025scheduling}, and TACO~\cite{wang2024taco}, the FLASQ bound~\cite{flasq}, and Litinski's analytic model~\cite{litinski2019game}.
DASCOT accepts only CX{+}$T$ circuits, so we compare it against \rush{} compiled on the \emph{same} CX{+}$T$ transpilations (matched $T$-count and measurement depth).

\myparagraph{Cost model}
We adopt FLASQ's cost model~\cite{flasq}.
A design occupies an $n{\times}m$ grid of distance-$d$ tiles at pitch $2(d{+}1)$~\cite{flasq}, one fabric of $N_{\mathrm{phys}} = 2\big(n(d{+}1)-1\big)\big(m(d{+}1)-1\big) - 1$ physical qubits.
Each tile suffers a per-cycle logical error 
\begin{equation*}
    p_{\mathrm{cyc}}(d) = c_{\mathrm{cyc}} (p_{\mathrm{th}}/p_{\mathrm{phys}})^{-(d+1)/2},
\end{equation*}
with $p_{\mathrm{th}}{=}10^{-2}$, $p_{\mathrm{phys}}{=}10^{-3}$, $c_{\mathrm{cyc}}{=}0.03$~\cite{flasq}.
A program with $M$ magic states and measurement depth $D$ runs for $L$ rounds, with $t_r(d) = t_r/(d\,t_{\mathrm{cyc}})$ the reaction latency in rounds ($t_{\mathrm{cyc}}{=}1\,\mu$s, $t_r{=}10\,\mu$s).
A single shot takes $W = t_{\mathrm{cyc}}\,d\,L$ and fails with the first-order budget
\(
\varepsilon = d\,p_{\mathrm{cyc}}(d)\,\big(V + M\,t_r(d)\big)
              + p_{\mathrm{mag}} M,
\)
where $V$ is the exposed tile-rounds of live logical information and $p_{\mathrm{mag}}$ the end-to-end error of one cultivated magic state.
Magic states are priced from the cultivation dataset shipped with FLASQ: one cultivated $T$ costs its expected space--time volume $v_{\mathrm{cult}}(d)$ in blocks, ${\sim}3.0$ at $d{=}13$ and ${\sim}2.0$ at $d{=}15$ (a \emph{block} is one tile occupied for one logical round) 
and carries an error of $p_{\mathrm{mag}} = 9.7\times10^{-7}$~\cite{gidney2024magic,flasq}.
Patch motion via walking qubits~\cite{mcewen2023relaxing} is priced at the resting rate. An optional penalty $p_{\mathrm{mv}}$ charges the tiles a patch crosses at $p_{\mathrm{cyc}}$ evaluated at $p_{\mathrm{phys}}{+}p_{\mathrm{mv}}$ (\tabref{sensitivity} tests up to $p_{\mathrm{mv}} = 2p_{\mathrm{phys}}$).

A shot succeeds with probability $P_{\mathrm{succ}} = e^{-\varepsilon}$, repeating until success takes $T_{\mathrm{succ}} = W / P_{\mathrm{succ}}$, and an operating point is \emph{feasible} iff $\varepsilon < 1$, our cutoff for the regime where FLASQ's first-order expansion is valid~\cite{flasq}.

\begin{table}[t]
    \centering\footnotesize
    \begin{tabularx}{\columnwidth}{l*{6}{>{\centering\arraybackslash}X}}
\toprule
system & min chip & $T_{\mathrm{succ}}$ & qubits & log.\ vol.\ & \mbox{qubit-s} & bound gap \\
\midrule
liblsqecc & \cellcolor[HTML]{FCF3F3}$1.7\times$ & \cellcolor[HTML]{F3CFD0}$9.2\times$ & \cellcolor[HTML]{FBF0F1}$2.0\times$ & \cellcolor[HTML]{F2CCCD}$10.7\times$ & \cellcolor[HTML]{F1C7C9}$13.2\times$ & \cellcolor[HTML]{F2CBCC}$55.0\times$ \\
LSQCA & \cellcolor[HTML]{FEFBFB}$1.2\times$ & \cellcolor[HTML]{F5D7D8}$6.4\times$ & \cellcolor[HTML]{FCF2F2}$1.8\times$ & \cellcolor[HTML]{F3D0D1}$8.8\times$ & \cellcolor[HTML]{F2CCCD}$10.8\times$ & \cellcolor[HTML]{F4D1D2}$40.2\times$ \\
O3LS & \cellcolor[HTML]{FCF2F2}$1.8\times$ & \cellcolor[HTML]{F5D5D6}$6.9\times$ & \cellcolor[HTML]{F9E7E7}$3.1\times$ & \cellcolor[HTML]{F5D5D6}$7.1\times$ & \cellcolor[HTML]{F4D4D4}$7.5\times$ & \cellcolor[HTML]{F5D7D8}$31.1\times$ \\
DASCOT & \cellcolor[HTML]{F8E4E4}$3.5\times$ & \cellcolor[HTML]{F9E6E6}$3.2\times$ & \cellcolor[HTML]{F6D9DA}$5.8\times$ & \cellcolor[HTML]{F7E0E1}$4.2\times$ & \cellcolor[HTML]{F6DBDC}$5.2\times$ & -- \\
Litinski & \cellcolor[HTML]{FCF4F5}$1.6\times$ & \cellcolor[HTML]{F7E0E1}$4.1\times$ & \cellcolor[HTML]{FBF0F1}$2.0\times$ & \cellcolor[HTML]{FBEFEF}$2.1\times$ & \cellcolor[HTML]{F8E3E3}$3.7\times$ & \cellcolor[HTML]{FAEAEA}$13.1\times$ \\
TACO & \cellcolor[HTML]{FDF9F9}$1.3\times$ & \cellcolor[HTML]{FBEEEE}$2.2\times$ & \cellcolor[HTML]{FDF5F5}$1.6\times$ & \cellcolor[HTML]{FAE9E9}$2.8\times$ & \cellcolor[HTML]{FAE9E9}$2.8\times$ & \cellcolor[HTML]{F9E9E9}$13.7\times$ \\
PureMagic & \cellcolor[HTML]{FBEEEE}$2.2\times$ & $0.9\times$ & \cellcolor[HTML]{FEFAFA}$1.3\times$ & \cellcolor[HTML]{FDF8F8}$1.4\times$ & $0.9\times$ & $4.7\times$ \\
\rowcolor[HTML]{E4EEF8}\textbf{RushHour} & $1.0\times$ & $1.0\times$ & $1.0\times$ & $1.0\times$ & $1.0\times$ & $4.8\times$ \\
\bottomrule
\end{tabularx}

    \caption{\textbf{Per-system summary.} Median ratios to \rush{} and geomean
        gap to the FLASQ~\cite{flasq} bound. DASCOT~\cite{molavi2025dependency} is measured on CX{+}$T$ inputs and compiles $\numDascotCoverage$ of $\numBenchmarks$ benchmarks.}
    \label{tab:summary}
\end{table}

\myparagraph{Metrics}
We report physical qubits, $T_{\mathrm{succ}}$, logical volume $nm\,L$, physical volume $N_{\mathrm{phys}} W$ (\emph{qubit-seconds}), the budget $\varepsilon$ by mechanism, and compile time.
Chip budgets are stated relative to each circuit's \emph{physical minimum}, the smallest chip on which the FLASQ bound itself runs the circuit. Layout headroom is \emph{tiles per qubit}. The \emph{dependency floor} is the circuit's critical path with \CNOT{}, $T$, and $S$ at one logical round each and Hadamards free.

\textit{Per-baseline pricing conventions and validation details are in Appendix~\ref{app:fairness}.}
\section{Evaluation}
\label{sec:evaluation}

\begin{figure}[t]
    \centering
    \includegraphics{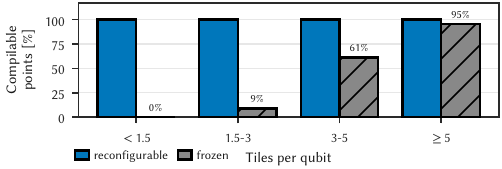}
    \caption{\textbf{Reconfiguration ablation} (\secref{eval:space}).
        \textit{Freezing removes a median $\numFrozenRemoved\%$ of operating points (up to $\numFrozenRemovedMax\%$).}}
    \label{fig:frozen}
\end{figure}

We evaluate whether \rush{} runs efficiently across constrained and abundant space (\textbf{RQ1}), reaches successful execution faster under matched space (\textbf{RQ2}), and spans the full space--time trade-off at the best resource efficiency (\textbf{RQ3}). We further assess its sensitivity to hardware characteristics, the cost of dynamism, the contribution of individual components, and compiler scalability (\textbf{RQ4}).
\tabref{summary} summarizes the main results.

\begin{figure*}[t]
    \centering
    \includegraphics{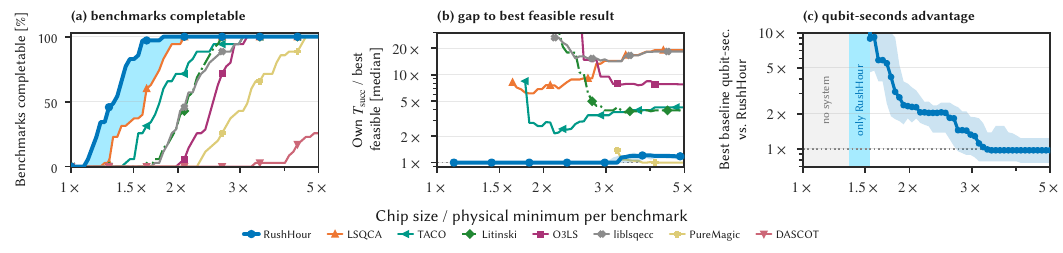}
    \caption{\textbf{Comparison across chip sizes} (\secref{eval:space}, \secref{eval:time}).
        (a)~Suite share completed, cyan: only \rush{} runs. (b)~Median gap to the
        best feasible $T_{\mathrm{succ}}$. Infeasible counts as unbounded, so
        curves exist where a system completes half the suite. Band: CI95.
        (c)~Qubit-seconds of the best feasible baseline over \rush{} (median,
        quartiles). \textit{\rush{} completes the suite at $\numSuiteBudgetRushhour\times$ the
            physical minimum, the best baseline at $\numSuiteBudgetBaseline\times$. Its gap to the best
            feasible result never exceeds ${\sim}\numMaxGapToBest\times$.}}
\label{fig:advantage}
\end{figure*}

\begin{figure*}[t]
    \centering
    \includegraphics{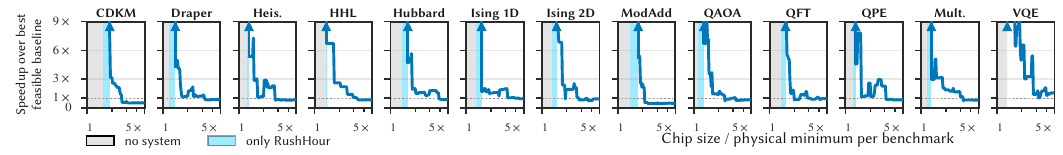}
    \caption{\textbf{Per-family speedup across chip budgets}
        (\secref{eval:time}). Per benchmark family, the median member's speedup
        over the best feasible baseline. Gray: no system runs, cyan: only
        \rush{} (speedup unbounded). The axis clips at $9\times$ and curves
        enter from above.}
    \label{fig:family-regimes}
\end{figure*}

\subsection{Space Efficiency}
\label{sec:eval:space}

\EQ{feasibility}{Feasibility}{Can \rush{} run circuits on small chips that are infeasible for static approaches?}
Comparing minimum feasible chips per circuit (\figref{advantage}a), we find each design needs a $\numMinChipLo\times$ (LSQCA) to $\numMinChipHi\times$ (DASCOT) larger chip than \rush{}, and at \rush{}'s minimum chip $\numSoleCount$ of $\numBenchmarks$ benchmarks run on no baseline.
On the tightest chips ($1.5\times$ the physical minimum), \rush{} runs $\numTightChipRushhour$ of $\numBenchmarks$ circuits where the best static design (LSQCA) runs $\numTightChipStatic$, and where any baseline is feasible at all, \rush{} reaches a successful shot a median $\numSpeedupTight\times$ sooner.

\EQ{spaceimprovement}{Space improvement}{How much space does \rush{} save at matched time?}
To match \rush{}'s time, every static design needs $\numQubitsMatchedLo$--$\numQubitsMatchedHi\times$ \rush{}'s physical qubits, even when each is taken at its fastest operating point (\tabref{summary}).

\begin{figure}[t]
    \centering
    \includegraphics{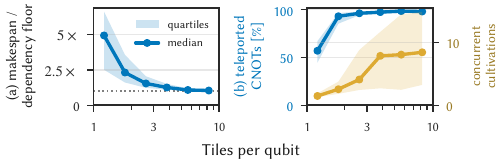}
    \caption{\textbf{Overhead hiding and concurrent work} (\secref{eval:time}).
        Against layout headroom. \textit{(a)~From five tiles per qubit the schedule
            sits within $\numFloorHeadroom\times$ of the floor. (b)~Teleportation serves $\numTeleportedShare\%$ of
            \CNOT{}s, and $\numCultivationsLo$--$\numCultivationsHi$ magic states cultivate in parallel.}}
    \label{fig:hiding}\label{fig:dynamism-cost}
\end{figure}

\EQ{reconfig}{Reconfiguration ablation}{Is dynamic reconfiguration the cause of this feasibility?}
We freeze dynamic reconfiguration by making any operating point whose schedule requires reconfiguration infeasible.
\figref{frozen} resolves the loss by layout headroom: below $1.5$ tiles per qubit, not a single frozen operating point compiles, and only $\numFrozenSurvivesTight\%$ survive at $1.5$--$3$ tiles.

\rqa{1}{\rush{} runs on a median $\numMinChipLo$--$\numMinChipHi\times$ smaller minimum feasible chip than every baseline, and $\numSoleShare\%$ of benchmarks execute where no static design does. Freezing reconfiguration removes a median $\numFrozenRemoved\%$ of operating points.}

\begin{figure*}[t]
    \centering
    \includegraphics{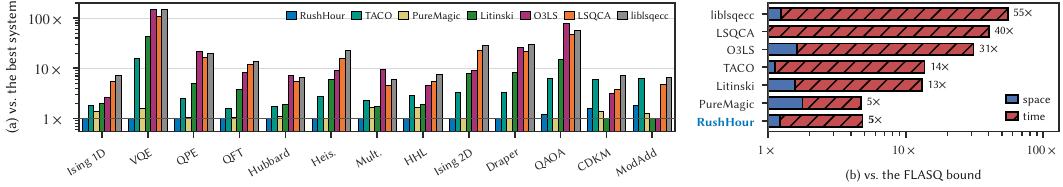}
    \caption{\textbf{Qubit-seconds volume comparison}
        (\secref{eval:tradeoff}). (a)~Relative qubit-seconds per benchmark family. (b)~Factored gap to the FLASQ bound. \textit{\rush{} is the cheapest system on \numFamiliesCheapest{} of \numFamilies{} families and within $2\%$ on two more, and lies only $\numBoundGapGeomean\times$ from the FLASQ bound.}}
    \label{fig:tradeoff-bars}
\end{figure*}

\begin{figure}[t]
    \centering
    \includegraphics{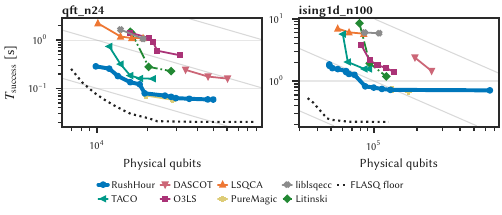}
    \caption{\textbf{Space--time frontiers} (\secref{eval:tradeoff}). \textit{Static designs are isolated points or short frontiers. One \rush{}
            architecture spans the full range and reaches chips no baseline fits.}}
    \label{fig:pareto}
\end{figure}

\subsection{Time Efficiency}
\label{sec:eval:time}

\EQ{matched}{Matched budgets}{At the same physical-qubit budget, how much sooner does \rush{} reach a successful shot?}
\figref{family-regimes} and \figref{advantage}b resolve the comparison across chip budgets, per benchmark family and in aggregate.
On space-constrained chips, \rush{} reaches a successful shot up to a median $\numSpeedupTight\times$ sooner than the best feasible alternative, and a median $\numMatchedLo$--$\numMatchedHi\times$ sooner than every design except PureMagic at matched qubits (\tabref{summary}).
With abundant space, \rush{} runs on par with PureMagic: on the spacious chips PureMagic compiles to, PureMagic runs a median $\numMatchedPuremagic\times$ \rush{}'s time at matched qubits (\tabref{summary}), while \rush{} matches its time at space parity and stays within ${\sim}\numMaxGapToBest\times$ of the best feasible result at every budget (\figref{advantage}b).

Where compilers lead in time, they are using a Pauli-based-computation (PBC) model in which Clifford absorption compresses the input circuit~\cite{hofmeyr2025scheduling,zhu2026o3ls,litinski2019game,bravyi2016trading}. O3LS leads by up to $\numPbcLeadOthreels\times$ on four serial adders, PureMagic by up to $\numPbcLeadPuremagic\times$ on \numPbcLeadCountPuremagic{} circuits at a large space overhead, and Litinski's blocks by up to $\numPbcLeadLitinski\times$ on five.
PureMagic's programs run at $\numPuremagicFloorRatio\times$ the circuit-model dependency floor, and replayed without PBC optimization, its schedules lose to \rush{} by a median $\numNoPbcReplay\times$.
\rush{} itself executes at $\numFloorHeadroom\times$ its circuit-model dependency floor once the layout has headroom, near-optimal within its model.

\EQ{hiding}{Overhead mitigation}{Can \rush{} successfully hide the overhead dynamism could induce?}
With layout headroom, the makespan sits within $\numFloorHeadroom\times$ of the dependency floor. Below two tiles per qubit, the median is $\numFloorTightLo\times$, rising to $\numFloorTightHi\times$ at the tightest layouts, against $\numFloorStaticPooled\times$ for the tightest static design (\figref{hiding}).
Dynamic rotations make Hadamards effectively free: even packed below two tiles per qubit, only one in ${\sim}\numDynamismHadamardsPerRotationTight$ Hadamards requires a physical rotation, and none at all beyond five (\tabref{dynamism}).
A median ${\sim}\numCultivationsLo$ cultivations run at once on packed layouts and ${\sim}\numCultivationsHi$ beyond three tiles per qubit, saturating at the parallelism the circuit's dependencies admit, which keeps magic states ready while the makespan holds the floor.

\rqa{2}{On space-constrained chips \rush{} reaches a successful shot a median $\numSpeedupBandLo$--$\numSpeedupBandHi\times$ sooner than the best feasible alternative, and its median gap to the best feasible result never exceeds ${\sim}\numMaxGapToBest\times$ at any budget. With layout headroom, its schedule sits within $\numFloorHeadroom\times$ of the dependency floor.}

\subsection{Trade-off: Spanning the Frontier}
\label{sec:eval:tradeoff}

\EQ{frontier}{Frontier width}{Does one \rush{} architecture span the space--time trade-off?}
We record the entire (physical qubits, $T_{\mathrm{succ}}$) frontier against every baseline (\figref{pareto}).
Per circuit, \rush{}'s frontier holds a median of $\numFrontierRushhour$ non-dominated operating points, while each static design, swept over the same code distances, contributes a median of two to five and never more than eight.

\EQ{volume}{Space--time volume}{How does \rush{} compare in space--time volume?}
On space-constrained chips, the cheapest feasible baseline pays a median $\numQubitSecondsAdvantage\times$ (up to over $\numQubitSecondsAdvantageMax\times$) \rush{}'s qubit-seconds at $2\times$ the physical minimum, and the advantage settles to parity only at the sweep's spacious end (\figref{advantage}c). Per family, \rush{} is the cheapest system on $\numFamiliesCheapest$ of $\numFamilies$ and on par at $\numFamiliesWithinTwoPercent$ families (\figref{tradeoff-bars}a).

\EQ{optimum}{Optimality}{How close to the idealized limit is \rush{}?}
Against the FLASQ bound (\figref{tradeoff-bars}), \rush{}'s geomean per-circuit gap is $\numBoundGapGeomean\times$ and never exceeds $\numBoundGapMax\times$. \rush{} is on par with PureMagic, which only runs on large chips. Every other bound-priced design is at $\numBoundGapOthersLo$--$\numBoundGapOthersHi\times$.
\rush{} and PureMagic reach parity by opposite routes (\figref{tradeoff-bars}b): \rush{} runs closer to the bound in space ($\numBoundSpaceRushhour\times$ against $\numBoundSpacePuremagic\times$), PureMagic closer in time ($\numBoundTimePuremagic\times$ against $\numBoundTimeRushhour\times$) through space-excessive PBC-based computations. LSQCA pays ${\sim}\numBoundTimeLsqca\times$ in time for dense but routing-starved layouts.

\EQ{metric}{Logical volume}{How does \rush{} score on logical volume?}
\tabref{summary} lists logical volume (tiles${\times}$rounds) alongside qubit-seconds, each system at its own best point under each metric: \rush{} leads on this metric too, with PureMagic nearest at $\numLogicalVolumeBest\times$ and every other design at $\numLogicalVolumeNext\times$ or more.

\begin{figure}[t]
    \centering
    \includegraphics{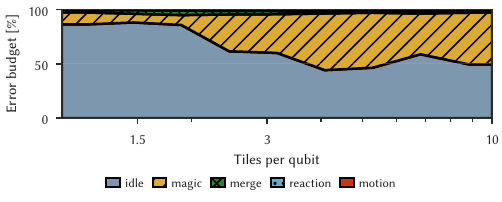}
    \caption{\textbf{Per-shot error budget} (\secref{eval:practicality}).
        Composition against layout headroom. \textit{Idle exposure peaks at $\numErrorIdleMax\%$ when
            tight. Magic-state error reaches parity near four tiles per qubit.
            Merges stay below $\numErrorMergeBound\%$, motion below $\numErrorMotionBound\%$.}}
    \label{fig:error-budget}
\end{figure}

\rqa{3}{One \rush{} architecture spans the entire space--time frontier with a median $\numFrontierRushhour$ operating points against the static designs' $\numFrontierBaselineLo$--$\numFrontierBaselineHi$, running $\numBoundGapGeomean\times$ from the FLASQ bound. On space-constrained chips it is the most resource-efficient by far, with the cheapest feasible alternative paying a median $\numQubitSecondsAdvantage\times$ (up to over $\numQubitSecondsAdvantageMax\times$) its qubit-seconds, and in spacious regimes it stays on par with the best existing design.}

\subsection{Practicality}
\label{sec:eval:practicality}

\begin{table}[t]
    \centering\scriptsize
    \setlength{\tabcolsep}{3.5pt}
    \begin{tabularx}{\columnwidth}{lll*{4}{>{\centering\arraybackslash}X}}
\toprule
$t_r$ & $p_{\mathrm{ph}}$ & $p_{\mathrm{mv}}$ & $T$ [$\times$] & qubit-s [$\times$] & gap [$\times$] & cover.\ [$\times$] \\
\midrule
$1\,\mu$s & $10^{-3}$ & 0 & 1.2--10.9 & 1.2--14.5 & \cellcolor[HTML]{EEBBBD}9.5 (11.8) & 1.8 (2.5) \\
\rowcolor[HTML]{E4EEF8}$10\,\mu$s & $10^{-3}$ & 0 & 0.9--9.2 & 0.9--13.2 & \cellcolor[HTML]{F3CFD0}4.8 (4.7) & 1.8 (2.1) \\
$100\,\mu$s & $10^{-3}$ & 0 & 0.8--5.4 & 1.3--7.1 & \cellcolor[HTML]{FAEBEB}2.0 (2.8) & 1.6 (2.1) \\
$1\,\mu$s & $5{\cdot}10^{-4}$ & 0 & 0.9--8.2 & 0.9--7.2 & \cellcolor[HTML]{EFBDBE}8.9 (9.0) & 2.3 (2.1) \\
$10\,\mu$s & $5{\cdot}10^{-4}$ & 0 & 0.7--7.7 & 0.8--7.4 & \cellcolor[HTML]{F3D0D1}4.7 (4.0) & 2.3 (2.1) \\
$100\,\mu$s & $5{\cdot}10^{-4}$ & 0 & 1.0--6.0 & 1.4--6.6 & \cellcolor[HTML]{FBEEEE}1.8 (2.9) & 1.6 (1.8) \\
$10\,\mu$s & $10^{-3}$ & $p_{\mathrm{ph}}$ & 0.9--9.2 & 0.9--13.2 & \cellcolor[HTML]{F3CFD0}4.9 (4.7) & 2.0 (2.1) \\
$10\,\mu$s & $10^{-3}$ & $2p_{\mathrm{ph}}$ & 0.9--8.9 & 0.9--13.1 & \cellcolor[HTML]{F3CECF}5.2 (4.7) & 2.4 (2.1) \\
\multicolumn{3}{l}{2 rounds/slide (proven rate)} & 0.9--9.2 & 0.9--13.2 & \cellcolor[HTML]{F3CFD0}4.9 (4.7) & 1.8 (2.1) \\
\bottomrule
\end{tabularx}

    \caption{\textbf{Hardware sensitivity.} Matched $T_{\mathrm{succ}}$,
        qubit-seconds, bound gap, and coverage across hardware settings. Gap and coverage list \rush{} with the best baseline in parentheses.}
    \label{tab:sensitivity}
\end{table}

\EQ{cost}{Cost of dynamism}{What does the dynamism itself cost, in error budget and patch motion?}
At the tightest layouts, magic-state error takes $\numErrorMagicTight\%$ of the budget and merges $\numErrorMergeTight\%$ behind the dominant idle exposure. With headroom, magic-state error grows toward parity with more tiles per qubit (\figref{error-budget}).
Reshapes are issued only where a gate cannot otherwise be served: even packed below two tiles per qubit they stay at one per ${\sim}\numDynamismGatesPerReshapeTight$ ancilla-served gates, patch motion holds $\numDynamismMotionShareTight\%$ of occupied space--time, and beyond five tiles every mechanism drops to exactly zero (\tabref{dynamism}).

\begin{table}[t]
    \centering\footnotesize
    \begin{tabularx}{\columnwidth}{l*{3}{>{\centering\arraybackslash}X}}
\toprule
tiles per qubit & $<2$ & $2$--$5$ & $\geq 5$ \\
\midrule
reshapes per anc.-served gate & 1/140 & 1/11k & $0$ \\
rotations per Hadamard & 1/1,035 & 1/102k & $0$ \\
patch motion [\% space-time] & $0.14\%$ & $<0.01\%$ & $0$ \\
\bottomrule
\end{tabularx}

    \caption{\textbf{Cost of dynamism} (\secref{eval:practicality}).}
    \label{tab:dynamism}
\end{table}

\begin{table}[t]
    \centering\footnotesize
    \begin{tabularx}{\columnwidth}{l*{2}{>{\centering\arraybackslash}X}}
\toprule
removed policy & $T_{\mathrm{succ}}$ vs.\ full & reshapes vs.\ full \\
\midrule
teleportation & \cellcolor[HTML]{FDF5F5}$1.05\times$ ($1.6\times$) & \cellcolor[HTML]{E79EA0}$1.35\times$ ($613\times$) \\
lookahead selection & \cellcolor[HTML]{FEFCFC}$1.00\times$ ($1.2\times$) & \cellcolor[HTML]{F3CED0}$1.10\times$ ($10\times$) \\
placement portfolio & \cellcolor[HTML]{FEF9F9}$1.04\times$ ($1.3\times$) & -- \\
\bottomrule
\end{tabularx}

    \caption{\textbf{Policy ablations} (\secref{eval:practicality}). Median
        (worst-case) ratio to the full compiler in $T_{\mathrm{succ}}$ at matched
        qubits and in reshapes.}
    \label{tab:policies}
\end{table}

\EQ{sensitivity}{Hardware sensitivity}{Does \rush{}'s advantage survive different reaction times, error rates, and motion costs?}
We re-price every baseline, the bound, and \rush{} at reaction times of $1$--$100\,\mu$s, physical error rates of $10^{-3}$ and $5{\cdot}10^{-4}$~\cite{acharya2025below}, and motion penalties up to $2p_{\mathrm{phys}}$ (\tabref{sensitivity}).
At $t_r{=}1\,\mu$s routing dominates every schedule and the gaps widen (matched speedups $\numSensitivityFastMatchedLo$--$\numSensitivityFastMatchedHi\times$). At $t_r{=}100\,\mu$s the reaction chain dominates and the field converges toward parity (per-system medians $\numSensitivitySlowMatchedLo$--$\numSensitivitySlowMatchedHi\times$).
Walking is priced at the resting rate $p_{\mathrm{mv}}{=}0$ (\secref{background}).
Even so, at $p_{\mathrm{mv}}{=}p_{\mathrm{phys}}$ the time and qubit-second advantages are unchanged and \rush{} still covers the suite on the smallest chip budget. Doubling the penalty widens the bound gap from $\numSensitivityBoundGap\times$ to $\numSensitivityBoundGapDoubleMotion\times$.
At the two-round walking rate, results move by at most $\numSensitivityWalkingShift\%$~\cite{mcewen2023relaxing}.

\EQ{policies}{Compiler policies}{Which compiler policies carry the results?}
Removing policies one at a time from the otherwise identical compiler (\tabref{policies}) shows that teleportation carries the schedule: removing it costs up to $\numPolicyTeleportTimeMax\times$ in $T_{\mathrm{succ}}$ on chain-dominated circuits and up to $\numPolicyTeleportReshapeMax\times$ in reshapes, while removing lookahead costs up to $\numPolicyLookaheadReshapeMax\times$ in reshapes.
The placement portfolio beats committing to the wrong single member by up to $\numPolicyPlacementMax\times$ (\tabref{policies}).

\EQ{compiletime}{Compile time}{Does the compiler scale to full circuits and full candidate sweeps?}
\figref{compile-cdf} shows the wall-clock time of every compile of every sweep.
A circuit's full candidate sweep is a median of $\numSweepCompiles$ independent compiles, resulting in a median of $\numSweepCpuMinutes$ CPU-minutes in total, entirely parallelizable.
DASCOT, which searches for dependency-optimal schedules, compiles only $\numDascotCoverage$ of the $\numBenchmarks$ circuits within the same one-hour budget~\cite{molavi2025dependency}.
This indicates that \rush{}'s dynamic approach admits a simple compiler that nevertheless matches or beats far more expensive searches such as DASCOT's.

\rqa{4}{\rush{}'s dynamism is effectively free. \rush{}'s advantage widens at faster reaction times. \rush{} compiles a median candidate in $\numCompileMedianSeconds\,$s.}

\section{Related Work}
\label{sec:related}

\myparagraph{Resource-estimation models}
FLASQ~\cite{flasq} models an idealized machine that allocates ancilla space fluidly and prices gates at optimistic resource usage. Litinski's model~\cite{litinski2019game} estimates runtime for a static PBC-based \LS{} architecture. LeBlond and Bennink compare that Clifford-eliminating family against direct Clifford+$T$ compilation on Hamiltonian-simulation workloads~\cite{leblond2026comparison}. We use FLASQ as the idealized limit (\secref{methodology}) and show a constructive compiler operating within $\numBoundGapGeomean\times$ of it across constrained and spacious regimes.

\myparagraph{Lattice-surgery compilers on fixed layouts}
O3LS~\cite{zhu2026o3ls}, DASCOT~\cite{molavi2025dependency},
LSQCA~\cite{kobori2025lsqca}, TopoLS~\cite{zhou2026topols}, and
liblsqecc~\cite{watkins2024high,leblond2024realistic} optimize placement, routing, and
scheduling against a floorplan that is fixed before execution. O3LS
searches over candidate layouts and LSQCA relocates qubits within its
fixed memory--compute floorplan, as earlier approaches do~\cite{hua2021autobraid, silva2024multi, kan2025sparo}. PureMagic~\cite{hofmeyr2025scheduling} reassigns which fixed tiles produce magic states or serve as routing per step, and TACO~\cite{wang2024taco} co-designs a Clifford-eliminating transpiler with a fixed layout tailored to the result. 
Ecmas+~\cite{zhu2025ecmasplus} customizes chip initialization per circuit and reaches depth-optimal schedules in its sufficient-resources regime by assuming magic states arrive freely at the data patches. \rush{} pays for their cultivation on-lattice and, given comparable layout headroom, still runs within $\numFloorHeadroom\times$ of its dependency floor.
Recent work orthogonally improves static designs using bounded-depth space--time routing~\cite{hamada2026bounded} or workload-aware floorplan selection at compile time~\cite{ghosh2026workload}.

\begin{figure}[t]
    \centering
    \includegraphics{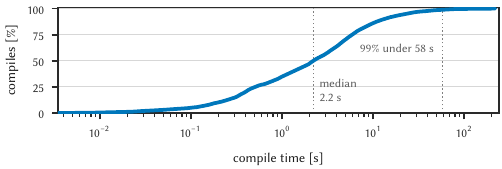}
    \caption{\textbf{Compile wall-clock time} (\secref{eval:practicality}) over
    every compilation. \textit{Half finish within $\numCompileMedianSeconds$\,s and
    $99\%$ within a minute.}}
    \label{fig:compile-cdf}
\end{figure}

\myparagraph{Mobile logical qubits}
Runtime movement of logical patches is an emerging direction~\cite{sharma2025spacetime,herzog2025movable,mcewen2023relaxing}. Sharma and Murali~\cite{sharma2025spacetime} densify early-FT layouts by moving data patches one tile per logical cycle. LSQCA moves qubits between memory and compute zones at high cost (${\sim}\numBoundTimeLsqca\times$ in time, \secref{eval:tradeoff})~\cite{kobori2025lsqca}. Herzog et al.\ exploit the fact that, on the color code, an \LS{} \CNOT{} can optionally teleport a qubit as it executes~\cite{herzog2025movable}.
These approaches move individual data patches, whereas \rush{} reconfigures free ancilla space across the lattice in constant time.

\myparagraph{Walking qubits}
Walking qubits have been used to handle leakage errors~\cite{mcewen2023relaxing} and dense memory architectures~\cite{low2026denser,gidney2025yoked}.
\rush{} is the first constructive approach to systematically use walking qubits for more resource-efficient space--time trade-offs.

\section{Conclusion}
\label{sec:conclusion}
We present \rush{}, a dynamic \LS{} architecture that enables resource-efficient execution in both constrained and spacious regimes. \rush{} dynamically reconfigures ancilla space, allocates resource states locally and just-in-time, and adapts qubit placement and orientation to span the space--time trade-off in a single framework.
The \rush{} ISA formalizes this execution model and ensures valid reconfiguration and allocation, the Lattice Management Unit (LMU) manages the evolving lattice, and the \rush{} compiler efficiently pipelines operations.

\rush{} enables running on previously infeasible chips with $\numSoleShare\%$ of benchmarks executing on no baseline at all, and every existing design needing a $\numMinChipLo$--$\numMinChipHi\times$ larger chip. On space-constrained chips where alternatives run at all, \rush{} is both faster and leaner, reaching a successful shot up to a median $\numSpeedupTight\times$ sooner and paying a median $\numQubitSecondsAdvantage\times$ (up to over $\numQubitSecondsAdvantageMax\times$) less in qubit-seconds, while on large chips it performs on par with the state of the art.
These results indicate that \rush{} successfully implements dynamic \LS.

\ifdefined\ready

\else
\myparagraph{Artifact} \rush{} will be publicly available for artifact evaluation, along with the entire experimental setup.
\fi

\ifdefined\ready
\section*{Acknowledgments}
We thank William Huggins for helpful discussions.
This work was funded by the Bavarian State Ministry of
Science and the Arts as part of the Munich Quantum Valley
(MQV) initiative, grant number 6090181.
\fi

\bibliographystyle{ACM-Reference-Format}
\bibliography{refs}

\ifdefined\ready
\appendix
\section{Fairness and Validation Details}
\label{app:fairness}

\myparagraph{TACO}
TACO's transpiler runs as released. The published design is a single compute block ($1.5n{+}4$ tiles), which we sweep over one, two, and four blocks ($1.5n{+}4k$) and report at its best per-circuit configuration.
TACO is rescheduled based on per-distance producer occupancy, so, like \rush{}, it pays the cultivation price in both schedule time and the error budget.

\myparagraph{LSQCA}
LSQCA follows its paper's closed-form geometries, reproducing the per-block tile and round counts and the floorplan formulas it reports.
It runs over its published configuration family, point- and line-SAM floorplans over the published bank counts (up to four for line-SAM, two for point-SAM), with its paper's native 15-beat magic-state factories.
Its exposure is footprint${\times}$makespan, an upper bound its tight memory density keeps close.

\myparagraph{DASCOT}
DASCOT runs as released on its CX{+}$T$ inputs. Its schedules are lowered into the \rush{} ISA and priced through the same replay, making DASCOT a second independent producer of ISA programs.

\myparagraph{liblsqecc}
liblsqecc~\cite{watkins2024high} runs as released, in its latest version with improvements described by LeBlond et al.~\cite{leblond2024realistic}. It runs in its non-clogging compact layout at its best per-circuit configuration, and its exposure is priced on the active volume its slicer reports.

\myparagraph{O3LS}
O3LS ships no artifact, so we reimplement it. The reimplementation reproduces its published reductions over its paper's SPC baseline on its own benchmark suite (time steps within $1\%$ on the standard layout, footprint within $2$ percentage points).

\myparagraph{Reaction latency}
Reaction latency is charged uniformly by what each schedule resolves. For \rush{}, PureMagic, and DASCOT we replay the schedules and measure the reaction-exact path, for which $t_r(d)\,D$ is a lower bound. TACO stalls in schedule, so its published rounds already carry the wait. The reaction-blind designs (Litinski, LSQCA, liblsqecc, O3LS) are charged the serial tail, $L = L_0 + t_r(d)\,D$ over their scheduled makespan $L_0$~\cite{flasq,fowler2012time,battistel2023realtime}.

\myparagraph{PureMagic}
PureMagic runs from its released artifact over its published Pauli-product weight limits $\omega \in \{1, \infty\}$, with cultivation at the same per-distance occupancy \rush{} pays.
Its schedules are replayed under the same reaction-exact critical-path rule \rush{} applies to itself: a decode lag of $t_r(d)$ rounds is charged on every dependency edge that leaves a $T$ merge's target or magic patch, and only there, plus the final $T$'s trailing decode. Every product keeps its artifact-scheduled start as a floor, preserving PureMagic's ordering, routing, and overlaps, and stall idle is charged exactly as \rush{} charges itself.

\myparagraph{Litinski blocks}
The Litinski blocks follow their paper's closed-form geometries. They are an analytic resource model.
Their exposure charges the full block during scheduled steps and only the data patches during the additive reaction tail, where no merge is in flight.

\myparagraph{Pricing conventions}
Every baseline pays the cultivation price in the error budget and the amortized footprint, while keeping its published gate ordering.
\rush{}'s exposure is its schedule's occupied tile-rounds (idle, merge, motion, and cultivation), the same basis as each baseline whose artifact reports occupancy. A baseline without such reporting pays footprint${\times}$makespan.
The feasibility cutoff is applied uniformly. As a sensitivity check, we also re-admit every baseline operating point that fails the cutoff, pricing it through the same formula $T_{\mathrm{succ}} = W e^{\varepsilon}$. Even then, the best baseline remains a median $\numEpsilonGrantMedian\times$ behind on the chips where \rush{} is otherwise the only feasible system.
We re-ran the comparison under each alternative convention we could construct, including TACO's own supply accounting and five further assumptions in its favor, and \rush{}'s median advantage at matched qubits never fell below $\numTacoFairnessFloor\times$.

\section{Benchmark Suite}
\label{app:per-benchmark}

\begin{table}[!htb]
    \centering\footnotesize
    \begin{tabular*}{\columnwidth}{@{\extracolsep{\fill}}lll@{}}
\toprule
benchmark & $n$ & $T$-count \\
\midrule
CDKM ripple-carry adder & 16, 32, 64, 80, 96 & 56--376 \\
Draper QFT adder & 24, 48, 64, 80 & 20,965--181,150 \\
Heisenberg, 1D chain & 36, 64 & 73,502--132,162 \\
HHL & 12 & 108,893 \\
Fermi--Hubbard, 1D chain & 18, 72 & 15,495--58,760 \\
Ising, 1D chain & 36, 100 & 15,300--42,500 \\
Ising, 2D lattice & 64, 100 & 40,640--63,500 \\
Modular adder & 32, 64, 80, 96 & 120--376 \\
QAOA & 24, 64, 80 & 13,563--139,228 \\
QFT & 24, 64, 80, 96 & 9,491--46,787 \\
QPE & 32, 96 & 29,871--121,451 \\
Multiplier & 16, 40 & 19,756--186,197 \\
VQE & 32, 96 & 7,980--31,710 \\
\bottomrule
\end{tabular*}

    \caption{\textbf{Benchmark circuits.} The evaluated suite.}
    \label{tab:benchmarks}
\end{table}

\tabref{benchmarks} shows the benchmark circuits used to evaluate \rush{}.

\clearpage

\fi

\end{document}